\documentclass[11pt]{article}
\usepackage{amsfonts,amsthm,amssymb,color}
\usepackage{framed,multicol,enumitem}
\usepackage{fancybox}
\usepackage{url}
\usepackage{mathtools}
\usepackage{algorithm, algpseudocode,amsmath,bm}
\usepackage[usenames,dvipsnames,svgnames,table]{xcolor}
\definecolor{darkgreen}{rgb}{0.0,0,0.9}
\usepackage{thmtools}
\usepackage[colorlinks=true, pagebackref,
citecolor=OliveGreen,linkcolor=BrickRed,urlcolor=BrickRed,pdfstartview=FitH]{hyperref}

\usepackage{fullpage}
\DeclareMathOperator{\image}{Im}
\def\be{{\bf e}}
\newtheorem{theorem}{Theorem}[section]

\newtheorem{lemma}[theorem]{Lemma}

\newtheorem{fact}[theorem]{Fact}

\newcommand{\suppress}[1]{}
\theoremstyle{definition}
\newtheorem{definition}[theorem]{Definition}

\newenvironment{fminipage}%
  {\begin{Sbox}\begin{minipage}}%
  {\end{minipage}\end{Sbox}\fbox{\TheSbox}}

\newenvironment{algbox}[0]{\vskip 0.2in
\noindent 
\begin{fminipage}{6.3in}
}{
\end{fminipage}
\vskip 0.2in
}

\def\C{\mathbb{C}}
\def\R{\mathbb{R}}
\def\Z{\mathbb{Z}}
\def\bz{{\bf z}}
\def\bx{{\bf x}}
\def\by{{\bf y}}

\def\bkappa{{{\bm{\kappa}}}}
\def\balpha{{\bm{\alpha}}}

\def\bone{{\bf 1}}

\def\bb{{\bf b}}

\def\pp{{\bf p}}

\newcommand\yy{\boldsymbol{\mathit{y}}}
\def\zz{{\bf z}}
\def\xx{{\bf x}}
\def\yy{{\bf y}}

\newcommand{\st}{\mbox{\rm s.t. }}

\newcommand{\Ex}{\mathbb{E}}

\newcommand{\coord}[3]{(#1_#2)_#3}

\begin{document}
	
\title{Nash Social Welfare for Indivisible Items under \\
Separable, Piecewise-Linear Concave Utilities}

\author{
{Nima Anari} \thanks{Stanford University,
	Email: {\sf anari@stanford.edu}.
}
\and
{ Tung Mai} \thanks{Georgia Tech,
	Email: {\sf tung.mai@cc.gatech.edu}.
}
\and
{ Shayan Oveis Gharan} \thanks{University of Washington,
	Email: {\sf shayan@cs.washington.edu}.
}
\and 
{ Vijay V. Vazirani}\thanks{Georgia Tech,
Email: {\sf vazirani@cc.gatech.edu}.
}
}

\date{}
\maketitle

\begin{abstract}
Recently Cole and Gkatzelis \cite{CG-Nash} gave the first constant factor approximation algorithm for the problem of 
allocating indivisible items to agents, under additive valuations, so as to maximize the {\em Nash social welfare (NSW)}. 
We give constant factor algorithms for a substantial generalization of their problem -- to the case of separable, piecewise-linear concave utility functions. We give two such algorithms, the first using market equilibria and the second using the
theory of stable polynomials.

In AGT, there is a paucity of methods for the design of mechanisms for the allocation of indivisible goods and the result of \cite{CG-Nash} seemed to be  taking a major step towards filling this gap. Our result can be seen as another step
in this direction.
\end{abstract}

\pagenumbering{gobble}

\newpage
	
\setcounter{page}{0}
\pagenumbering{arabic}

\section{Introduction}
\label{sec.intro}
It is well known that designing mechanisms for allocating indivisible items is much harder than for divisible items. In a sense, this dichotomy holds widely in the
field of algorithm, e.g., consider the difference in complexity of solving linear programs vs integer linear programs.
This difference is most apparent in the realm of computability of market equilibria, where even though the first
result introducing market equilibria to the AGT community, namely \cite{DPS}, dealt with the case of indivisible goods, there is a paucity of results for this case.
We are only aware of \cite{CR2007,BHM}; the first deals with smooth Fisher markets, in which small changes in
prices cause only proportionately small changes in demand, and the second studies
the question of allocating indivisible resources using the notion of competitive equilibrium from equal incomes (CEEI), giving 
algorithms and hardness results for two classes of valuations, namely perfect substitutes and perfect complements.
On the other hand, very impressive progress has been made for the case of divisible goods, e.g., see \cite{DPS,DPSV,DV1,JMS,JainAD,Chen.plc,VY,DK08,CSVY,GMSV,GMV,ChenTeng,CPY}.

Recently Cole and Gkatzelis \cite{CG-Nash} took a major step towards developing methodology for designing mechanisms for the 
allocation of indivisible items. They gave the first constant factor approximation algorithm for the problem of 
allocating indivisible items to agents, under additive valuations, so as to maximize the {\em Nash social welfare (NSW)}. 

%In an instance of the Nash welfare maximization problem 
They studied the following problem. We are given  a set of $m$ indivisible items  and we want to assign them to $n$ agents. An allocation vector is a vector  $\bx\in\{0,1\}^{[n]\times [m]}$ such that for each item $i$, exactly one $x_{a,i}$ is 1. Perhaps, the simplest model for the utility of an agent is the linear model. That is, each agent $a$ has a non-negative utility $u_{a,i}$ for an item $i$ and the utility that $a$ receives  for an allocation $\bx$ is
$$ u_a(\bx)=\sum_{i=1}^m x_{ai}u_{ai}.$$
The NSW objective is to compute an allocation $\bx$ that maximizes the geometric mean of agents' utilities,
$$ \left(\prod_{a=1}^n u_a(\bx)\right)^{\frac{1}{n}}.$$
%Note that conventionally NSW objective is defined as  the geometric mean of the agents' values. Here, we use the product to simplify our notation.

The above objective naturally encapsulates both fairness and efficiency and has been extensively studied as a notion of fair division (see ~\cite{Moulin,CKMPSW16} and references therein).
Cole and Gkatzelis \cite{CG-Nash} designed a $2e^{1/e}$ approximation algorithm for the above problem. This was later improved independently to $e$
in \cite{AOSS17} and $2$ in \cite{C+}.
%As is well known, NSW provides a happy compromise between fairness and efficiency, and is a well studied solution concept.
%We describe their techniques in Section \ref{sec.tech}. 

The case of indivisible goods is clearly very significant in AGT and there is a need to develop our understanding of such problems, both in terms of
positive and negative results. It is therefore natural to study generalizations of the Cole-Gkatzelis setting. 
Clearly, linear utility functions are too restrictive.
In economics, concave utility functions occupy a special place because of their generality and because they capture the
natural condition of decreasing marginal utilities. Since we wish to study allocation of indivisible items, we will assume that utility functions are 
piecewise-linear concave and additively separable over item types. In this paper, we obtain a constant factor approximation algorithm 
for NSW under these utilities -- this is a substantial generalization of the problem of \cite{CG-Nash}.

\suppress{
Suppose we have $m$ item {\em types}. For item type $i$, assume  we have a supply of $k_i$ units. The utility of each agent is separable over item types, but over each item type it is piecewise-linear concave. Let $u_{aij}$  be the marginal utility that agent $a$ receives from the $j$-th copy of item $i$. 
We assume that for all agents $a$ and items $i$,
$$ u_{ai1} \geq u_{ai2} \geq \dots \geq u_{aik_i}.$$	
	In an assignment, assume that agent $a$ receives $\ell_i$ copies of item $i$ for all $1\leq i\leq m$. Then her utility will be
	\[ \sum_{i=1}^{m}\sum_{j=1}^{\ell_i}u_{aij}. \]
	As before, the goal is to maximize the geometric mean of utilities of all agents.
}

The study of computability of market equilibria
started with positive results for the case of linear utility functions \cite{DPS,DPSV,DV1,JMS,JainAD}. However, its generalization to separable, piecewise-linear concave  
(SPLC) utilities was open for several years before
it was shown to be PPAD-complete \cite{ChenTeng,Chen.plc,VY}. Our first belief was that NSW under SPLC utilities should not admit a constant factor algorithm and that
the resolution of this problem lay in the realm of hardness of approximation results. Therefore, our positive result came as
a surprise.
We give constant factor approximation algorithms for our problem using two very different techniques.

\subsection{Problem Formulation}
\label{sec.formulation} 
Assume that there are $n$ agents and $m$ item types. For item type $i$, assume that we have a supply of $k_i$ units. The utility of each agent is separable over item types, but over each item type it is piecewise-linear concave.
	
	Now define $u_{aij}$ to be the marginal utility that agent $a$ receives from the $j$-th copy of item $i$.
	For each agent $a$ and item type $i$ we assume
	$$u_{ai1} \geq a_{ai2} \geq \dots \geq u_{aik_i} \geq 0$$

An allocation vector $\bx$ is a  vector where for each item type $i$,
 $$\sum_{a,j} x_{aij} \leq k_i.$$ %is equal to 1 for exactly one agent and zero otherwise. 
 We say $\bx$ is an integral allocation vector if all coordinates of $\bx$ are $0$ or $1$.	
 In other words, we allocate at most $k_i$ copies of each item type $i$. 
 For an allocation vector $\bx$, the utility of agent $a$ is 
$$u_a(\bx)=\sum_{i=1}^m \sum_{j=1}^{k_i} x_{aij} u_{aij}.$$ 
%to denote the utility of agent $a$ for this allocation. 

	% If in an assignment agent $a$ receives $\ell_i$ copies of item $i$, his utility will be
%	\[ \sum_{i=1}^{m}\sum_{j=1}^{\ell_i}u_{aij}. \]
	For clarity of notation throughout this paper  we use $\prod_a u_a(\bx)$ to denote the Nash welfare of an allocation $\bx$.
	With this notation, the goal is to maximize the product of utilities of all agents.
	This problem can be captured by the following integer program:
	\begin{equation}
		\label{eq:defintprog}
		\begin{aligned}
			\max_{x_{aij}} ~~~ & \left( \prod_{a=1}^{n} \sum_{i=1}^{m}\sum_{j=1}^{k_i} x_{aij}u_{aij}\right)^{1/n}, &\\
			\st ~~~ & \sum_{a=1}^{n}\sum_{j=1}^{k_i} x_{aij}\leq k_i & \forall i\\
			& x_{aij}\in \{0,1\} & \forall a, i, j\\
		\end{aligned}
	\end{equation}
%Note that the set of feasible solutions of the above mathematical program are all allocation vectors $\bx$.
	
%\section{Preliminaries}

\subsection{Contributions}
Our emphasis in this paper is on the development of techniques for designing mechanisms for the allocation of indivisible items. We prove our main theorem using two different techniques. The first one exploits the structure of the market equilibrium, building on \cite{CG-Nash}.
\begin{theorem}\label{thm:market}
The spending restricted algorithm given in Figure~\ref{alg:main} runs in polynomial time and yields a fractional allocation which when rounded using
the algorithm in Figure \ref{alg:marketRounding}, gives a factor 2 approximation algorithm for NSW for SPLC utilities.
\end{theorem}

The factor 2 algorithm for the linear case was shown to be tight in \cite{C+}. Hence, the bound for our algorithm stated above is also tight.

Our second approach is purely algebraic and uses the machinery of real stable polynomials, building on \cite{AOSS17,AO16}.
\begin{theorem}\label{thm:mainstable}
Program \eqref{eq:convprog} is a convex relaxation of the Nash-welfare maximization problem with SPLC utilities. There is a randomized  algorithm  that rounds any feasible solution of the convex program  to an integral solution with the Nash welfare at least $1/e^2$ fraction of the optimum (of \eqref{eq:convprog}) in expectation.
\end{theorem}

\suppress{
We remark that our second theorem is not algorithmic at this moment. Although we analyze a randomize rounding algorithm and we show that its expected welfare is at least $1/e^2$ fraction of the optimum,  the welfare of our algorithm can have a variance which is exponentially large in $n$. In principal, one should be able to use the method of conditional expectation to turn our algorithm into a polynomial time algorithm. But, that requires an approximate counting algorithm for a certain generalization of the result of \cite{JSV04} on counting the number of perfect matching in bipartite graphs. We leave this question as an open problem. 
%First, we use a market equilibrium based idea which is along the lines of the work of Cole and Gkatzelis. Recently, \cite{} obtained
%a constant factor algorithm for the case of additive valuations using results from the theory of stable polynomials. Our second algorithm
%extends their approach to SPLC  utility functions. Both methods involve new ideas and these are detailed in Section \ref{sec.tech}.
}

\subsection{Techniques}
\label{sec.tech}
In this part we discuss the main new ideas behind our two proofs. We start with the market equilibrium result.

\paragraph{Techniques Used in Market Equilibrium Approach.}
The starting point of \cite{CG-Nash}, henceforth called the CG-result, was a new market model in the Fisher setting called the {\em spending restricted model}.
 They modified the Fisher model as follows: each buyer has \$1 and the amount of money that can be spent on any good is at most
\$1, regardless of its price. 
As a result, the amount of good $i$ sold is $\min(1, 1/p_i)$, assuming unit amount of each good in the market. 
Moreover, since the total spending on any good can be at most \$1, the price of a highly desirable good is pushed higher 
and at equilibrium, each high priced good (having price 1 or more) is (essentially) allocated to one buyer. Clearly utilities of buyers can be scaled arbitrarily. 
At this point, the CG algorithm scales utilities of all buyers so that their utility from their maximum bang-per-buck goods equals the equilibrium
prices. Now, the product of the prices of high price items is an upper bound on OPT and the remaining problem is only to assign the low price
items integrally -- in the equilibrium, they have been allocated fractionally. This is done via rounding.

\cite{CG-Nash} gave a combinatorial polynomial time algorithm for finding an equilibrium in the spending restricted model by modifying 
combinatorial algorithms for finding equilibria for Fisher markets under linear
utility functions \cite{DPSV,Orlin}.
As mentioned above, equilibrium computation for SPLC utilities is not in P and so this starting point is not
available to us. The closest thing available is that if one assumes perfect price discrimination, then a polynomial time algorithm was
given by \cite{GV}. However, our problem has little to do with price discrimination. 
Our first task therefore was to define a suitable market model that will compute a fractional allocation which provides an upper bound
on the optimal NSW and to which rounding can be applied.

Our key clue comes from the observation that equilibrium prices of goods provide a proxy 
for the utility accrued by the agent who gets this good. The dilemma is that our market has multiple items of each type
and the same agent may derive different utility from different items of the same type.
Thus, our market model should sell different items of the same type for different prices!
At this point, it was natural for us to define a market model in which agents pay for utility rather than the amount of items they receive. 
In our model we impose a spending restriction of \$1 on each item of each type.

In our new market model buyers have SPLC utility functions and have unit money each. There is a base price for each type of item. If an agent buys more than 
one item of the same type, he spends the base price for the item that he receives the least utility from. For all other items of this type that he buys, the price of 
the item for him is scaled up so that the ratio of price/utility is the same.
If the price of type $i$ is $p_i$, then the amount of an item of this type that is sold at equilibrium is $\min (1, 1/p_i)$. 
%  what does price mean?
We call our model the {\em utility allocation market}, since buyers pay for the amount of utility they accrue,
in a certain well defined way (see Section \ref{sec.model}). 

Computing an equilibrium in this model is not straightforward. Our algorithm iteratively computes better and better approximations to the equilibrium
via scaling, using a parameter $\Delta$ which halves after each iteration. 
We maintain the invariant that the base price of each type can only increase. In each scaling iteration, buyers start with surplus money and
spend it all by the end of the iteration. The amount spent on each good can exceed its spending restriction, but by at most $\Delta$. 
The main difficulty in executing an iteration under SPLC utilities is that as base prices change, the money charged for different items of the same type changes
by different amounts because of the different utilities buyers accrue from them.

Our rounding algorithm is a generalization of the ones in \cite{CG-Nash,C+}. We first allocate to agents the integral part of their allocation. The new difficulty lies
in that the fractional allocations of a given item type add up to more than one item in general and there may not be a way of partitioning the fractions
so that they add up to a unit each. To get around this, we use information in the obtained fractional allocation to create a new instance under
linear utilities. We apply the rounding procedure of \cite{C+} to the latter instance and it yields an allocation for our SPLC problem. A critical step in this proof is
showing that the upper bound which \cite{C+} compare their solution to also applies to our SLPC instance.
Hence we get a factor 2 algorithm for the NSW problem with SPLC utilities.

\paragraph{Techniques Used in Real Stability Approach.}
Our second proof is purely algebraic and exploits the theory of real stable polynomials (see \cite{Pem12} for background). 
Consider an integral allocation vector $\bx$; recall that  the Nash welfare corresponding to this allocation is $\prod_a u_a(\bx)$. Note that this function is log concave in $\bx$ and it can be maximized by standard convex programming tools. Unfortunately, the ratio of the optimum of this convex program and the welfare of the optimum integral allocation can 
be unbounded  in the worst case \cite{CG-Nash}. 
%To get around this problem, in \cite{AOSS17}, the authors wrote the welfare as a polynomial of fictitious variables $y_1,\dots,y_m$. In this paper, we start with a generalization of this
%idea to SPLC utilities. Let
Let $\bx$ be an arbitrary fractional allocation, and suppose we allocate the items using the natural randomized rounding method: For each item type $i$, we generate $k_i$ samples $(a,j)$ independently where each sample is equal to $(a,j)$  with probability $x_{aij}/k_i$, for all $a,j$. We can study the (expected) welfare of this algorithm by summing up a subset of coefficients of a multivariate polynomial. Consider the polynomial
$$ p_\bx(y_1,\dots,y_m)=\prod_a \sum_{i=1}^m \sum_{j=1}^{k_i} x_{aij} u_{aij} y_i.$$
If $k_i=1$ for all $i$, then the (expected) welfare of the above randomized rounding algorithm is just the sum of the coefficients of multi-linear terms of the above polynomial. This is the main idea of \cite{AOSS17}.

In our case, however, the (expected) welfare of the randomized rounding algorithm (up to a normalizing factor and a loss of $e^n$) is the sum of the coefficients of all monomials of $p_\bx$ of degree $n$ where the degree of each $y_i$ is at most $k_i$ (see \autoref{lem:lb}).  
%For a polynomial $p\in \R[y_1,\dots,y_m]$ and a vector $\bkappa\in \Z_+^n$ let $C_p(\bkappa)$ be the coefficient of the monomial $\prod_i y_i^{kappa_i}$ of $p$. Let ${\cal M}$ be a matroid with elements $[m]$. 
%We employ a recent theorem \cite{AO16} which gives a 
%The key to decrease the integrality gap is to define the welfare of (the fractional allocation) $\bx$ as the 
As a sanity check, observe that  if $\bx$ was indeed integral only these monomials have nonzero coefficients in $p_\bx(y_1,\dots,y_m)$.
Therefore, the question boils down to writing a convex relaxation of the sum of coefficients of monomials of $p$ of degree $n$ where the degree of each $y_i$ is at most $k_i$. 
%We can write this as follows: Let 
%$$q(y_1,\dots,y_m)=\sum_{\bkappa\in \Z_+^n: \kappa_i\leq k_i, \forall i} \prod_i y_i^{\kappa_i}$$
%We need to write a convex relaxation for $\sum_{\bkappa\in \Z_+^m} C_{p_\bx}(\bkappa)C_q(\bkappa)$, where we use $C_p(\bkappa)$ to denote the coefficient of the monomial $\prod_i y_i^{\kappa_i}$ in $p$.

We use a recent result \cite{AO16} where it is shown that for any real stable polynomial $p$ and any subset $S$ of monomials of $p$ that ``correspond'' to a real stable polynomial, there is a convex function that approximates the sum of coefficients of monomials of $p$ in $S$ within a factor of $e^k$ where $k$ is the largest degree of monomials in $S$ (see \autoref{thm:pqstability} for more details). To be precise, consider the polynomial
$$ q(y_1,\dots,y_m)=\frac{\partial_t^n}{n!} \prod_{i=1}^n \sum_{j=0}^{k_i} t^j y_i^j\Big|_{t=0},$$
i.e., the coefficient of $t^n$ in $\prod_{i=1}^n \sum_{j=0}^{k_i} t^j y_i^j$.
This polynomial has all monomials of degree $n$ in $y_1,\dots,y_m$ such that the degree of each $y_i$ is at most $k_i$; the coefficient of each such monomial is 1. To put it differently, $q$ can be seen as the generating polynomial of the set of bases of a {\em laminar} matroid of depth 2.

The expected welfare of our randomized rounding algorithm is $q(\partial \by)p_\bx(\by)|_{\by=0}$, up to normalizations. Unfortunately, polynomial $q$ is not real stable (see 
Section~\ref{sec:stableprelim} for definition of real stability); so \autoref{thm:pqstability} is not applicable. %For example, $1+y_1+y_1^2$ ...
Instead, we work with a real stable polynomial $\tilde{q}(y_1,\dots,y_m)$ such that $\tilde{q}\approx q$. Let, 
$$ \tilde{q}(y_1,\dots,y_m)=\frac{\partial_t^{K-n}}{(K-n)!}\prod_{i=1}^m (t+y_i/k_i)^{k_i}\Big|_{t=0},$$
where $K=\sum_i k_i$; in other words, $\tilde{q}$ is the coefficient of $t^{K-n}$ in $\prod_{i=1}^m (t+y_i/k_i)^{k_i}$. 
It is not hard to see that the above polynomial is real stable. Furthermore,  the coefficient of each monomial of $\tilde{q}$ is at least $e^{-n}$, i.e., $\tilde{q}$ is an $e^n$ approximation of $q$ at any point. So, up to another loss of $e^n$ we can use \autoref{thm:pqstability} to find a fractional allocation vector $\bx$ that maximizes $\tilde{q}(\partial \by)p_\bx(\by)|_{\by=0}$; this gives our convex relaxation of the NSW with SPLC utilities. Then, it follows that the above randomized rounding algorithm receives at least $e^{2n}$ fraction of the optimum welfare in expectation.

\subsection{Overview of the Paper}

In Section \ref{sec.market} we give our algorithm using the market-based approach. 
In Section \ref{sec.model} we define our {\em utility allocation market model} in which buyers are charged according to the utility they accrue -- thus two items of the same 
type could end up costing different amounts not only to two different buyers but also to the same buyer. 
Next we give the spending restricted version of this model.
In Section \ref{sec.alg-market} we give a combinatorial polynomial time algorithm for finding an equilibrium in this market.
The output of this algorithm is a fractional allocation. In Section \ref{sec.upper} we show how to derive an upper bound on OPT from this allocation.
Finally, in Section \ref{sec.round-alg} we show how to round this solution, hence yielding a factor 2 approximation algorithm.

In Section \ref{sec.stability} we give our algorithm based on real stable polynomials.
Section \ref{sec:stableprelim} defines key concepts and states the relevant theorems of Gurvits and \cite{AO16}.
In Section \ref{sec:stableproof} we give our convex program for this approach, followed by an algorithm for rounding the fractional allocation
obtained by solving this convex program. We prove that this yields an $e^2$-approximation algorithm for our problem.

\section{The Market-Based Approach}
\label{sec.market}

%We start by defining our utility allocation market model with spending restrictions.

\subsection{Spending-Restricted Utility Allocation Market}
\label{sec.model}

\paragraph{Utility Allocation Market}
%We give a market-based interpretation of the convex program such that an equilibrium allocation of the market corresponds exactly to a solution of the convex program. 

We will assume there are $n$ agents and $m$ item-types in the market. Each type $i$ has a supply of $k_i$  items. The utility of each agent is additively separable over the item types, but piecewise-linear concave over each type of items. Each agent has a budget of one dollar, and each type $i$ has a base price $p_i$ for a single item of that type. 
The price of an item to an agent can be more than its base price as defined below.
We denote $(a,i,j)$ by the $j$-th item of type $i$ that $a$ receives (fractionally). 
 
%Consider an agent $a$ spending money on at least $j$ items of type $i$. Define the \emph{base bang-per-buck} of the $j$-th item with respect to $a$ to be $\frac{u_{aij}}{ p_i}$, and the %\emph{actual bang-per-buck} of the item with respect to $a$ to be $\frac{u_{aij}}{ p_i + q_{aij}}$. There are two ways $a$ can spend money on his $j$-th item of type $i$:
%\begin{enumerate}
%	\item Spending $x_{aij}p_i$ to obtain $x_{aij}$ of the item if he already bought $j-1$ full items of type $i$ and the actual bang-per-buck values of those $j-1$ items are equal to the %base bang-per-buck of the $j$-th item.
%	\item Spending $q_{aij}$ to decrease the actual bang-per-buck of the $j$-th item if he already paid $p_i$ to get the full item. This is necessary if he wants to receive more than $j$ %items of type $i$.
%\end{enumerate}

\begin{definition}
	The \emph{bang-per-buck} of an item for an agent $a$ is the ratio of utility $a$ derives from the item and the amount of money $a$ spends on the item.  
\end{definition}

\paragraph{Admissible Spending}  

%Given a set of prices, an admissible spending of agent $a$ at target bang-per-buck $b_a$ is given as follows: 
%\begin{enumerate}
%	\item For the $j$-th item of type $i$ whose base bang-per-buck with respect to $a$ is strictly greater than $b_a$, spend $p_i + q_{aij}$ to get the full item. Let $q_{aij} = %\frac{u_{aij}}{b_a} - p_i$ so that the actual bang-per-buck of the item is exactly equal to $b_a$. Stop if $a$ spends all of his budget.
%	\item For the $j$-th item of type $i$ whose base bang-per-buck with respect to $a$ is equal to $b_a$, spend $x_{aij}p_i$ to get $x_{aij}$ fraction of the item for $x_{aij} \leq 1$ and %$x_{aij}$ as large as possible. 
%\end{enumerate} 
%It can be seen that if $b_a$ decreases, the amount of money $a$ spends at target bang-per-buck $b_a$ cannot decrease. Therefore, if we set $b_a$ initially at a very large value and %decrease its value continuously, $a$ must spend all of his budget at some value of $b_a$. 
%We call an admissible spending of $a$ at such value of $b_a$ an \emph{admissible spending} of $a$, and call $b_a$ the \emph{bang-per-buck of $a$}. 

Given a set of prices, an agent can only spend his money in a certain way. Roughly speaking, each agent spends so that he gets optimal utility for the amount of money he spends at current prices. We define an \emph{admissible spending} of an agent $a$ as follows: 

\begin{definition}An \emph{admissible spending} of agent $a$ is a spending such that: 
\begin{enumerate}
	\item The money $a$ spends is at most his budget.
	\item There exists a value $b_a$ such that for each $(a,i,j)$:
		\begin{enumerate}
			\item If $b_a < u_{aij}/p_i$, $(a,i,j)$ is called a \emph{superior} item. In this case, $a$ must spend $\frac{u_{aij}}{b_a}$ to get the full item ($x_{aij} = 1$). 
			%We call 
			%\[ q_{aij} = \frac{u_{aij}}{b_a} - p_i \]
			%the \emph{extra utility spending} of $a$ on the item.
			\item If $b_a = u_{aij}/p_i$, $(a,i,j)$ is called an \emph{active} item. In this case, $a$ can spend $x_{aij} p_i$ to get $x_{aij}$ fraction of the item. 
			\item If $b_a > u_{aij}/p_i$, $(a,i,j)$ is called an \emph{inferior} item. In this case, $a$ must not spend on $(a,i,j)$ ($x_{aij} = 0$).
		\end{enumerate}

%	all items that $a$ spends money on has actual bang-per-buck of $b_a$ with respect to $a$. For all other items, their base bang-per-buck values must be less than $b_a$. Moreover, for each item whose base bang-per-buck with respect to $a$ is strictly greater than $b_a$, $a$ must spend $\frac{u}{b_a}$ to get the full item, where $u$ is the amount of utility $a$ derives from the item. Specifically, if the item is the $j$-th item of type $i$ that $a$ gets, then $q_{aij} = \frac{u_{aij}}{b_a} - p_i$, and the amount of money $a$ spends on the item is $p_i + q_{aij} =  \frac{u_{aij}}{b_a}.$
\end{enumerate}
%If $x_{aij} > 0$, we call $x_{aij} p_i$ the \emph{base spending} of $a$ on the item.
%If $a$ spends all of his budget, we call it a \emph{full admissible spending}.
The corresponding $b_a$ is called the \emph{bang-per-buck} of $a$.
\end{definition}
 
The equivalent notion of admissible spending in linear utility Fisher market is simply spending on items that maximize the ratio of utility to price. 
Moreover, with respect to an admissible spending,  the total amount of money $a$ spends on $(a,i,j)$
%the $j$-th item of type $i$ 
can be written as 
$p_{i}x_{aij} + q_{aij}$ where
\[ q_{aij}  = \begin{cases} \frac{u_{aij}}{b_a} - p_i  &\mbox{if $(a,i,j)$ is a superior item,} \\ 
0 & \mbox{otherwise. } \end{cases} \] 
We call $p_{i}x_{aij}$ the \emph{base spending}  and $q_{aij}$ the \emph{extra utility spending} of $a$ on the item.
 
The extra utility spending captures our idea that buyers are charged according to the utility they accrue. 
Specifically, the $q$-variables guarantee that the bang-per-buck value is the same for all items that $a$ receives. 
In other words, $a$ has to pay more for items that give him more utility, even though they might be of the same type.

\paragraph{Spending-Restricted Equilibrium} A price $\pp$ is \emph{spending-restricted equilibrium price} if 
there exists an admissible spending for each buyer such that the budgets are fully spent and the total base spending on each item of type $i$ is exactly $\min(p_i,1)$.
The corresponding allocation $\xx$ is called a \emph{spending-restricted allocation}.
 
Note that in this spending-restricted equilibrium, items with price greater than 1 do not have to be completely sold. 
Moreover, the total spending on each item is at most 1 (hence the name spending-restricted). To see this, consider an arbitrary item. 
If it is a superior item, there is only one agent spending on it. Therefore, the total spending on the item is at most the budget of the agent, i.e., 1.
Otherwise, if it is not a superior item, the total spending on it is equal to the base spending and bounded by 1.

\subsection{Algorithm to Compute an Equilibrium Solution}
\label{sec.alg-market}

\paragraph{High-level Idea of the Approach} Recall that our goal is to compute a price and an allocation such that:
\begin{enumerate}
	\item The spending of each agent is an admissible spending. 
	\item Each buyer spends all his budget.
	\item The total base spending on all items of type $i$ is $k_i \min(p_i,1)$.
\end{enumerate}
Following the approach of \cite{DPSV}, the idea of our algorithm is maintaining conditions 1 and 3 while satisfying condition 2 gradually. Note that condition 1 makes sure that the money spent by buyers is at most their budget. The algorithm maintains a price for each item type and a bang-per-buck value for each agent. As the algorithm progresses, prices are increased and bang-per-buck values are decreased so that buyers with surplus money have an opportunity to spend their remaining budgets,
without violating conditions 1 and 3. When the money of all buyers is completely spent, all three conditions are satisfied and the algorithm terminates. 

\paragraph{Scaling Technique}
To obtain a polynomial running time, we use a scaling technique similar to the ones in \cite{Orlin} and \cite{CG-Nash}. As a consequence, rather than maintaining the third condition exactly, we make sure that it is approximately satisfied. 	
Let $$
	p_i(\Delta)=
	\begin{cases}
	\lceil p_i / \Delta \rceil\ \Delta & \text{ if $p_i$ is not a multiple of $\Delta$, }\\
	p_i + \Delta  & \text{ otherwise.}
	\end{cases}
	$$
We give the following definition of 
%\emph{$\Delta$-admissible spending} and 
\emph{$\Delta$-allocation}. 

%\begin{definition}	Consider an admissible spending of agent $a$. The amount of money $a$ spends on the $(a,i,j)$ is $p_i x_{aij} + q_{aij}$. 
%If instead he spends $p_i(\Delta) x_{aij} + q_{aij}$ for all $i$ and $j$, we say that the spending of $a$ is a \emph{$\Delta$-admissible spending}.
%\end{definition}

\begin{definition}
An allocation $\xx$ is a \emph{$\Delta$-allocation} with respect to a price vector $\pp$ if the spending of each agent is 
%a $\Delta$-admissible spending. 
an admissible spending, and the total base spending on all items of type $i$ is 
%\begin{itemize}
%	\item equal to $k_i \min\{ 1, p_i(\Delta)\}$ if $p_i$ is not a multiple of $\Delta$, 
%	\item at least $k_i \min\{ 1, p_i \}$ and at most $k_i \min\{ 1, p_i(\Delta) \}$ if $p_i$ is a multiple of $\Delta$. 
%\end{itemize}
at least $k_i \min\{ 1, p_i \}$ and at most $k_i \min\{ 1, p_i(\Delta) \}$. 
We say that $\pp$ \emph{supports} a $\Delta$-allocation.
If the agents spend all their budgets in a $\Delta$-allocation, $\xx$ is called a \emph{full $\Delta$-allocation}.
\end{definition}

Our scaling algorithm maintains a $\Delta$-allocation at all steps for appropriate values of $\Delta$. Specifically, $\Delta$ must be a power of 2 and is halved each scaling phase as the algorithm processes. Note that as $\Delta$ gets smaller, the value of $p_i(\Delta)$ gets closer to $p_i$, and the approximate version of condition 3 gets closer to the exact version. One can show that within $O(K \log V_{\text{max}})$ where $V_{\text{max}} = \max_{a,i,j,i',j'} \{ \frac{u_{aij}}{u_{ai'j'}}\}$ and $K = \sum_{i=1}^m k_i$ is the total number of items, the value of $\Delta$ is small enough, and hence a full $\Delta$-allocation gives an exact solution.

\paragraph{The Network $N(\pp,\bb)$} 
A key ingredient of the algorithm lies in constructing and computing max-flow in a directed network which we call $N(\pp,\bb)$.
The first step of our construction of $N(\pp,\bb)$ is to fully assign all superior items. To be precise, for each triplet $(a,i,j)$ such that $\frac{u_{aij}}{p_i} > b_a$, we set $x_{aij} = 1$ and charge $a$ the amount $\frac{u_{aij}}{b_a}$. At the end of this step, let $e_a$ be the amount of money agent $a$ spends on superior items and $l_i$ be the number of superior items of type $i$. 

We then construct a directed network as follows. The network has a source $s$, a sink $t$ and vertex sets $A$ and $I$ corresponding to agents and item-types respectively. The source $s$ is connected to each agent $a$'s vertex via a directed edge of capacity $1 - e_a$. Let $c(p_i,\Delta) = \min \{ 1, p_i(\Delta)\}$. For each item-type $i$, there is an an edge from type $i$'s vertex to $t$ of capacity $(k_i - l_i)c(p_i,\Delta)$. Finally, for each active item, there is an edge from the corresponding agent to the type of the item of capacity $c(p_i,\Delta)$. 

All the active allocations are done by computing a maximum flow in this network. Specifically, the amount of flow from an agent $a$ to a type $i$ corresponds to the amount that $a$ spends on active items of type $i$.

\subsubsection{A Subroutine} 
\label{sec.subroutine}

We give a price-increase algorithm that takes as input a parameter $\Delta$, a price $\pp$ that supports a $\Delta$-allocation and the corresponding bang-per-buck $\bb$ of the allocation. 
The algorithm then returns a price which supports a full $\Delta$-allocation together with its bang-per-buck vector.
The algorithm is given in \autoref{alg:priceIncrease}. Note that $\Delta$ remains unchanged throughout the algorithm.

The first step of the algorithm computes a max-flow in $N(\pp,\bb)$. The amount of flow in the network together with the allocation of superior items gives a $\Delta$-allocation.
If the agents spend all of their budgets, we have a full $\Delta$-allocation. 
Step 2 of the algorithm returns the current price $\pp$ and allocation $\xx$ if that happens.
 
Step 3 finds a set  $X \subseteq A$ and $Y \subseteq B$ such that $\left( s \cup X \cup Y, t \cup (A \setminus X) \cup (B \setminus Y) \right)$ forms a min-cut with maximum number of vertices on $t$ side of the cut. 
Since it is a min-cut, all edges from $X$ to $I \setminus Y$ are saturated. 
Furthermore, all edges from $A \setminus X$ to $Y$ carry no flow, and all agents with surplus money are in $X$. 
Since the cut maximizes the number of vertices on the $t$ side, there is no \emph{tight set} $T \in Y$. 
We say that a set $T \in Y$ is \emph{tight} if in the current network $\Gamma(T) = S$, and the total capacity of edges in $(s,S)$ is at most the total capacity of edges in $(T,t)$. Here, $\Gamma(T)$ denote the set of agent-vertices in $X$ connected to $T$ through an edge in the network. Clearly, if there is a tight set $T \in Y$, the cut defined by $s \cup (X\setminus \Gamma(T)) \cup (Y \setminus T)$ 
must also be a min-cut with more vertices on the $t$ side. The following lemma gives a crucial observation about the two sets $X$ and $Y$. 

\begin{lemma} 
\label{lem:sets}
For all $y \in Y$, the edge $(y,t)$ is saturated. Furthermore, if the capacity of $(y,t)$ increases, some agents in $X$ can spend more money.
\end{lemma}
\begin{proof}
Let $y$ be an arbitrary vertex in $Y$. Define a reachable subgraph $R$ as follows:
\[
R = \left\{ v \in X \cup Y: \exists \text{ a directed path from $v$ to $y$ in the residual graph of $N(\pp,\bb) \setminus \{s,t\}$} \right\}.
\]
In other words, $R$ is the set of vertices in $X \cup Y$ that are reachable from $y$ via paths alternating between edges in the reverse direction and edges carrying flow in the forward direction. 

Let $R_X = R \cap X$ and $R_Y = R \cap X$. Since all edges from $A \setminus X$ to $Y$ carry no flow, the total flow from $s$ to $R_X$ is equal to the total flow from $R_Y$ to $t$. Furthermore, since $s \cup X \cup Y$ defines a min-cut with maximum number of vertices on the $t$ side, the total capacity of edges from $s$ to $R_X$ must be greater than the total capacity of edges from $R_Y$ to $t$. It follows that there must be an agent in $R_X$ with surplus money. 

Let $x \in R_X$ be an agent with surplus money. From the definition of $R$, there is a residual path from $x$ to $y$ in $N(\pp,\bb) \setminus \{s,t\}$. Therefore, if $(y,t)$ is not saturated, there exists a residual path from $x$ to $t$. 
This contradicts the fact that the flow in $N(\pp,\bb)$ is a maximum flow. By the same reason, if the capacity of $(y,t)$ increases, 
$x$ can take the chance to spend more money.
\end{proof}

The final step increases price of the item-types in $Y$ and decreases the bang-per-buck of agents in $X$ in proportion. The increase in the prices of the types in $Y$ can allow agents with surplus money in $X$ to spend their remaining budgets. As the prices increase, the following events might happen:
\begin{enumerate}
	\item An inferior item of type $i$ in $I \setminus Y$ may become active for agent $a$ in $X$. This can happen because the bang-per-buck values of agents in $X$ keep decreasing. If this event occurs, we add the corresponding edge from $a$ to $i$ with capacity $c(p_i,\Delta)$ to the network.
	\item A superior item of type $i$ in $Y$ may become active for agent $a$ in $A \setminus X$. This can happen because the prices of types in $Y$ keep increasing. If this event occurs, we add the corresponding edge from $a$ to $i$ with capacity $c(p_i,\Delta)$. Also, since the item is no longer a superior item, we need to adjust the capacity of edges $(s,a)$ and $(i,t)$ accordingly. 
	\item The capacity of a type $i$ in $Y$ may increase. This can happen if $p_i$ is less than 1. If this event occurs, some agents with surplus money in $X$ have a chance to spend their budget. As a consequence, some sets in $Y$ might go tight.
\end{enumerate}

\begin{figure}[ht]
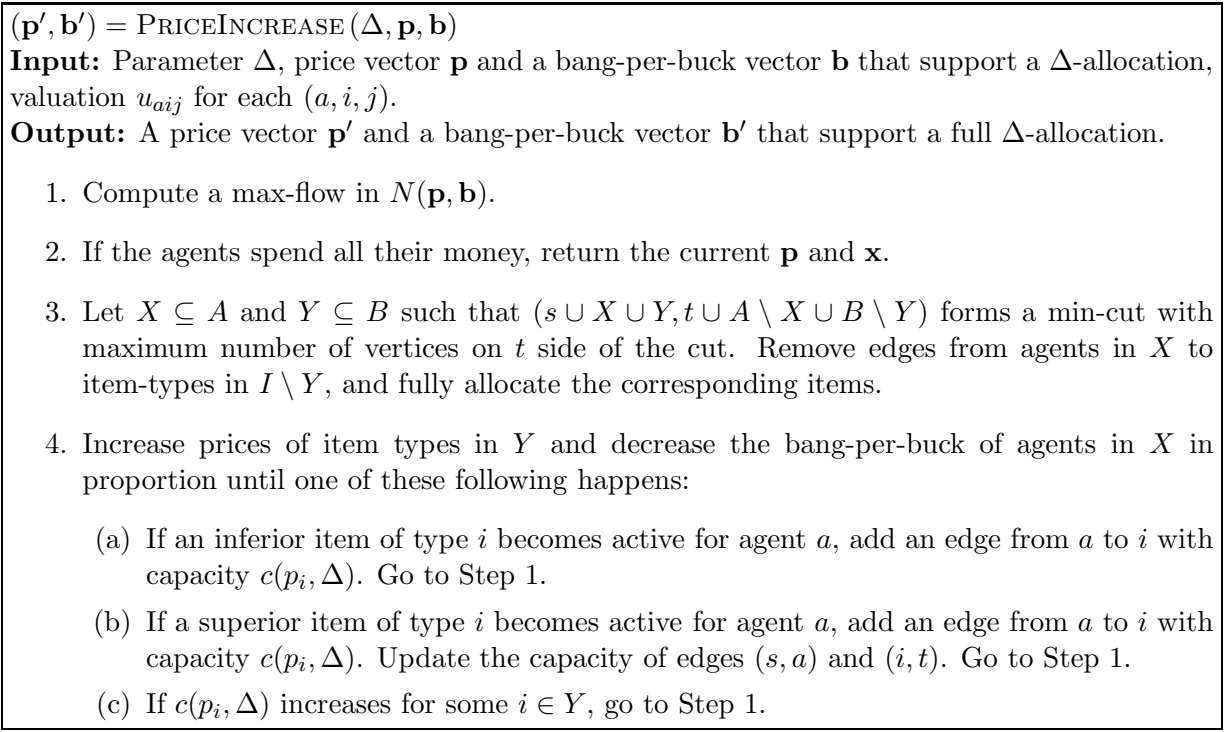

	
	\begin{algbox}
		$(\pp',\bb') = \textsc{PriceIncrease}\left( \Delta, \pp, \bb \right)$
		
		\textbf{Input:} Parameter $\Delta$,  price vector $\pp$ and a bang-per-buck vector $\bb$ that support a $\Delta$-allocation, valuation $u_{aij}$ for each $(a,i,j)$.
		
		\textbf{Output:} A price vector $\pp'$ and a bang-per-buck vector $\bb'$ that support a full $\Delta$-allocation.
		\begin{enumerate}
			\item Compute a max-flow in $N(\pp,\bb)$.
			\item If the agents spend all their money, return the current $\pp$ and $\xx$. 
			\item Let $X \subseteq A$ and $Y \subseteq B$ such that $( s \cup X \cup Y, t \cup A \setminus X \cup B \setminus Y)$ forms a min-cut with maximum number of vertices on 					$t$ side of the cut. Remove edges from agents in $X$ to item-types in $I \setminus Y$, and fully allocate the corresponding items.
			\item Increase prices of item types in $Y$ and decrease the bang-per-buck of agents in $X$ in proportion until one of these following happens: 
				\begin{enumerate}
					\item If an inferior item of type $i$ becomes active for agent $a$, add an edge from $a$ to $i$ with capacity $c(p_i,\Delta)$. Go to Step 1.
					\item If a superior item of type $i$ becomes active for agent $a$, add an edge from $a$ to $i$ with capacity $c(p_i,\Delta)$. Update the capacity 							of edges $(s,a)$ and $(i,t)$. Go to Step 1.
					\item If $c(p_i,\Delta)$ increases for some $i \in Y$, go to Step 1.
				\end{enumerate}
		\end{enumerate}
	\end{algbox}
	
	\caption{A Price-Increase Subroutine.}
	
	\label{alg:priceIncrease}
	
\end{figure}

\paragraph{Running Time} We give an upper bound on the running time of the algorithm in \autoref{alg:priceIncrease} as a function of the surplus money of the agents. 

\begin{lemma}
\label{lem:priceIncreaseRunningTime}
$\textsc{PriceIncrease}\left( \Delta, \pp, \bb \right)$ runs in poly$(K,s)$ time if the total surplus money of the agents is $s\Delta$.
\end{lemma} 

\begin{proof} Clearly, the first 3 steps of the algorithm run in polynomial-time. 
It is also easy to see that computing the prices and bang-per-buck values at which the next event happens requires polynomial-time. 
Therefore, it suffices to show that the number of events is polynomial in $K$ and $s$ as well.

First, consider an event in which the capacity $c(p_i,\Delta)$ increases for some type $i$. 
By \autoref{lem:sets}, if this event occurs, some agents with surplus money in $X$ have a chance to spend their remaining budget.
Therefore, $\Delta$ more budget is spent in the $\Delta$-allocation for the new prices, except the cases where some agents have less than $\Delta$ surplus budget. 
However, such cases can only happen at most $n$ times. 
Also, notice that the total amount of unspent budget can only go up when a superior item becomes active. 
However, in that case the increase in unspent budget can be spent in the next max-flow computation.
It follows that the number of capacity increase events is bounded by poly$(n,s)$.

Consider an event in which an inferior item becomes active. If this event occurs, the item will remain active for the corresponding agent until some capacities increase. Therefore, there can be at most poly$(K)$ such events per capacity increase event. A similar argument can be applied to the case in which a superior item becomes active. It follows that the total running time of $\textsc{PriceIncrease}$ is polynomial in $s$ and $K$.
\end{proof}

\paragraph{Correctness}
\begin{lemma}
\label{lem:priceIncreaseCorrectness}
$\textsc{PriceIncrease}\left( \Delta, \pp, \bb \right)$ returns a full $\Delta$-allocation with the corresponding price $\pp$ and bang-per-buck $\bb$.
\end{lemma} 

\begin{proof} Since Step 2 of the algorithm guarantees that it always terminates with the agents spending all of their money, it suffices to show that 
the algorithm maintains a $\Delta$-allocation at every step. 
To begin, the input price $\pp$ and bang-per-buck $\bb$ are given such that they support a $\Delta$-allocation. 
Therefore, constructing $N(\pp,\bb)$ and finding a maximum flow in this network give a $\Delta$-allocation. 

We will show that when each event in Step 4 happens, the current price $\pp$ and bang-per-buck $\bb$ still support a $\Delta$-allocation. 
Note that whenever we raise the price of items in $Y$, we also decrease the bang-per-buck of agents in $X$ by the same factor. 
Therefore, all edges from $X$ to $Y$ remain. 
Moreover, all edges that disappear are from $A \setminus X$ to $Y$ and carry no flow. 
It follows that the spending that was computed before we raised prices is still a valid spending.
 
If an inferior item becomes active, all the capacities remain unchanged, and the $\Delta$-allocation that the algorithm had before remains a $\Delta$-allocation.

%If a superior item of type $i$ becomes active, the capacities of edges $(s,a)$ and $(i,t)$ increase by the same amount of $p_i(\Delta)$, and the edge from $a$ to $i$ also have capacity of exactly $p_i(\Delta)$.
If a superior item of type $i$ becomes active, the capacity of $(s,a)$ increases by $p_i$, and the capacity of $(i,t)$ increases by $p_i(\Delta)$. 
Moreover, the edge from $a$ to $i$ also have capacity of exactly $p_i(\Delta)$.
Therefore, the next maximum flow computation will give a $\Delta$-allocation. 

Finally, if the capacity of some item-type $i$ increases, the algorithm still maintains a $\Delta$-allocation. 
To see this, notice that when this event occurs, $p_i$ is a multiple of $\Delta$, and $p_i + \Delta \leq 1$. 
Moreover, the capacity of $(i,t)$ increases by exactly $k_i' \Delta$ where $k_i'$ is the number of remaining items of type $i$.  
By \autoref{lem:sets}, the edge $(i,t)$ must be saturated with flow of value $k_i' p_i$ before its capacity increases. 
After the capacity increases, the flow on $(i,t)$ must be at least $k_i' p_i$ and at most $k_i' p_i(\Delta)$.
This value of flow corresponds to the amount of base spending on active items of type $i$.
Since the amount of base spending on each superior item of type $i$ is exactly $p_i$, the total base spending on all items of type $i$ is at least $k_i p_i$ and at most $k_i p_i(\Delta)$.
\end{proof}

\subsubsection{Polynomial Time Algorithm} 
\label{alg.poly}

In this section, we give a polynomial-time algorithm for computing a spending-restricted equilibrium (\autoref{alg:main}).

\begin{figure}[ht]
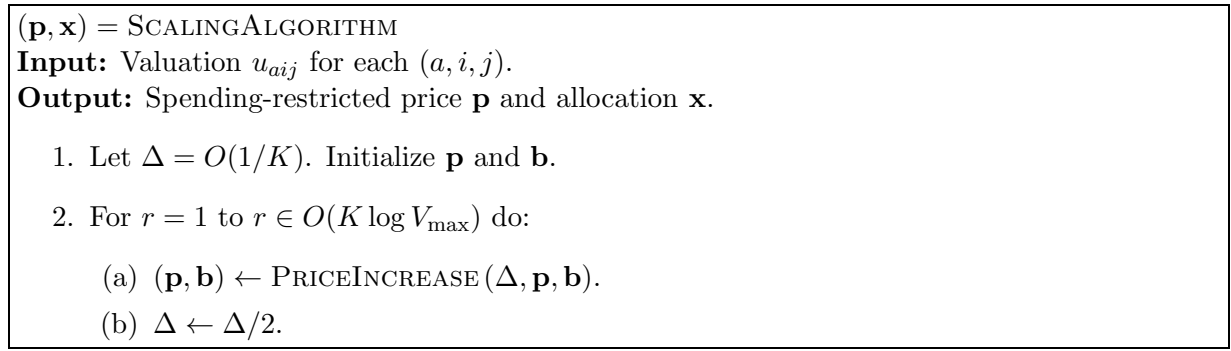

	
	\begin{algbox}
		$(\pp,\xx) = \textsc{ScalingAlgorithm}$
		
		\textbf{Input:} Valuation $u_{aij}$ for each $(a,i,j)$.
		
		\textbf{Output:} Spending-restricted price $\pp$ and allocation $\xx$.
		\begin{enumerate}
			\item Let $\Delta = O(1/K)$. Initialize $\pp$ and $\bb$.
			\item For $r = 1$ to $r \in O(K \log V_{\text{max}})$ do:
			\begin{enumerate} 
				\item $(\pp,\bb) \leftarrow \textsc{PriceIncrease}\left( \Delta, \pp, \bb \right)$.
				\item $\Delta \leftarrow \Delta / 2$.
			\end{enumerate}
		\end{enumerate}
	\end{algbox}
	
	\caption{A Polynomial Time Algorithm for Computing a Spending-Restricted Equilibrium.}
	
	\label{alg:main}
	
\end{figure}
The algorithm starts with $\Delta = O(1/K)$. Specifically, $\Delta$ is set to be the largest power of 2 that is at most $1/2K$.
 
The first step of the algorithm computes initial price $\pp$ and bang-per-buck $\bb$ together with a $\Delta$-allocation corresponding to these prices and bang-per-buck values. We will explain the details of this step later. 

The algorithm then repeatedly calls $\textsc{PriceIncrease}$ on the current $(\pp,\bb)$ and halves $\Delta$ in each scaling phase. Notice that when $\Delta$ is halved, the capacities of some edges 
%together with the spending allowed on superior items 
may decrease. 
As a result, some agents may have surplus money. 
However, the algorithm still maintains a $\Delta$-allocation with respect to the new $\Delta$.
After $O(K \log V_{\text{max}})$ scaling phases, the algorithm terminates with a full $\Delta$-allocation for $\Delta = O(2^{-K}/V_{\max})$.

\paragraph{Initialization} We initialize price $\pp$ and bang-per-buck $\bb$ for which there exists a $\Delta$-allocation. 
We assume that $\Delta$ is at most $1/2K$. 

To begin, we pick an arbitrary agent $i$ and find an appropriate $\pp$ and $b_i$ such that $i$ demands all the items that he derives positive utility from. 
This can be done by setting $p_i = u_{aik_{i}}/{M}$ for a large number $M$ and $b_a$ small enough. 
$M$ is chosen large enough so that $a$ only spends at most a half of his budget. 
Furthermore, the bang-per-buck of other agents is also set to be very large, and hence they only demand items of type $i$ such that $p_i = u_{aik_{i}} = 0$. 

If $a$ derives positive utility from all items, that is, $u_{aik_{i}} > 0$ for all $i$, then we are done since $a$ has enough money to pay an extra amount of $\Delta$ for each item.  Also, if there is no more demand on a type $i$ with zero price,  then we can leave $p_i = 0$ since items of type $i$ don't have to be fully sold. 
Therefore, we may assume that there is at least one agent $a'$ other than $a$ demanding an item of type $i$ with $p_i=0$. 
In this case, we raise $p_i$ by a small amount and set $b_{a'}$ so that $a'$ only demands (some of) the remaining items of type $i$. 

We then continue in this manner until all items with positive price are fully sold. The price $\pp$ supports a $\Delta$-allocation since all prices can be scaled to small values such that no agent spends more than $1/2$, and hence agents can pay an extra amount of $\Delta$ for each item.

\paragraph{Running Time}
\begin{theorem}
	$\textsc{ScalingAlgorithm}$ returns a full $\Delta$-allocation for $\Delta = O(2^{-K}/V_{\max})$ in polynomial time.
\end{theorem}
\begin{proof}
	Since the initial value of $\Delta$ is $O(1/K)$, by \autoref{lem:priceIncreaseRunningTime}, 
	the first call of $\textsc{PriceIncrease}$ takes polynomial time. 
	We will show that each subsequent call of $\textsc{PriceIncrease}$ runs in polynomial time as well. 
	\autoref{lem:priceIncreaseCorrectness} guarantees that at the beginning of each subsequent call of $\textsc{PriceIncrease}$ on parameter $\Delta$, the price $\pp$ and bang-per-buck $\bb$ support a $2\Delta$-allocation.
	When the value changes from $2\Delta$ to $\Delta$, for each $i$, $p_i(\Delta)$ can either decrease by $\Delta$ or remain unchanged.
	It follows that the total unspent budget can be at most $K \Delta$.
	\autoref{lem:priceIncreaseRunningTime} can be applied again to show that each subsequent call of $\textsc{PriceIncrease}$ terminates in polynomial time. 
\end{proof}

\subsection{Rounding a Spending-Restricted Solution}
\label{sec.round}

\subsubsection{Upper Bound on OPT}
\label{sec.upper}

Since scaling the valuations of the agents does not affect the solution of our problem, given a spending-restricted equilibrium price vector $\pp$, we can always scale the valuations of the agents such that the bang-per-buck of each agent from the equilibrium allocation is exactly 1. We say that the valuations of the agents are \emph{normalized} for $\pp$. The following lemma gives an upper bound for the NSW of the optimal integral solution based on spending-restricted prices. This is a generalization of the idea used in \cite{CG-Nash}. We say that an item of type $i$ is a \emph{high-price} item if $p_i > 1$ and a \emph{low-price} item otherwise. Also, we denote $H(\pp)$ by the set of high-price item types.
\begin{lemma}
\label{lem:upperbound}
	Let $\pp$ be a spending-restricted price vector and $\xx^*$ be an optimal integral solution. If the valuations of the agents are normalized for $\pp$ then
	\[ \left( \prod_a u_a(\xx^*) \right)^{1/n} \leq \left( \prod_{i \in H(\pp)} p_i^{k_i} \right)^{1/n}.\]
\end{lemma}
\begin{proof} First we give a bound on the sum of the agents' utilities in any allocation based on the spending-restricted price vector $\pp$. Consider a fractional spending-restricted allocation $\xx$ corresponding to $\pp$. Since valuations of the agents are normalized, each agent receives exactly 1 unit of utility in $\xx$. However, $\xx$ may not fully allocate the items to the agents, since high-price items may not be completely sold. Each one of the high-price items generates 1 unit of utility in $\xx$ since the total spending on it is precisely 1. Let $\yy$ be an allocation in which we allocate the rest of each high-price item of type $i$ to one of the agents spending on it, hence
generating $p_i - 1$ more utility. Therefore, the total utility that all the items generate in $\yy$ is: 
\[ n +  \sum_{i \in H(\pp)} ( p_i - 1)  =  n - \sum_{i \in H(\pp)} k_i + \sum_{i \in H(\pp)} k_i p_i. \]

We will show that the total utility of all the agents in any allocation cannot be larger than this. Consider the items of type $i$. In $\yy$, each one of those items generates either $p_i$ or more than $p_i$ utility. Moreover, any agent that can derive more than $p_i$ utility from an item actually receives the item in $\yy$. Therefore, $\yy$ allocates the items to the agents such that the total utility all the items generate is maximized. It follows that for any integral allocation $\zz$,
\[ \sum_a u_a(\zz) \leq \sum_a u_a(\yy) = n - \sum_{i \in H(\pp)} k_i + \sum_{i \in H(\pp)} k_i p_i.\]

Notice that $\sum_{i \in H(\pp)} k_i$ is the total number of high-price items. Therefore, 
$n - \sum_{i \in H(\pp)} k_i$
is a lower bound on the number of agents receiving only low-price items. Moreover,
in any integral allocation $\zz$, a high-price item must be assigned to only one agent.
 
We claim that the product $\prod_a u_a(\zz)$ is maximized if 
each high-price item of type $i$ generates $p_i$ utility, 
no agent receives more than one high-price item,
and all other agents who do not receive any high-price item derive exactly 1 unit of utility. 
In that case, the product $\prod_a u_a(\zz)$ is exactly $\prod_{i \in H(\pp)} p_i^{k_i}$.

Next, we prove the claim stated in the previous paragraph. 
We may assume that a high-price item of type $i$ generates $p_i$ utility since this can only increase $\prod_a u_a(\zz)$.
Suppose there is an agent $a$ receiving a high-price item together with one or more items. 
It follows that the total utility of agents receiving only low-price items is strictly less than $n - \sum_{i \in H(\pp)} k_i $.
Since there are at least $n - \sum_{i \in H(\pp)} k_i$ of them, there must be an agent $a'$ having less than 1 utility. 
Transferring all value but the high-price item from $a$ to $a'$ makes the product $\prod_a u_a(\zz)$ larger. 
Therefore, if an agent receive a high-price item, he should not receive anything else. 
For agents receiving only low-price items, their total utility is at most $n - \sum_{i \in H(\pp)} k_i$ and
the product is maximized if each of them derives exactly 1 unit of utility.
\end{proof}

\subsubsection{Rounding}
\label{sec.round-alg}	

We give a generalized version of the rounding procedure proposed in \cite{CG-Nash}. 

\paragraph{Algorithm} The first step of the of the rounding algorithm constructs a bipartite spending graph $G$ whose vertices are agents and individual items as follows.  
For each superior item of type $i$, add an edge between $i$ and the corresponding agent $a$, and let $a$ spend $p_i + q_{aij}$  on the item. Note that $a$ is the only agent spending on this item.  For the remaining items of type $i$, partition their allocation into units arbitrarily. Add edges between each unit and the agents who receive parts
of this unit. Note that the total spending on each item is at most 1 in $G$.

After that, we invoke the rounding procedure of \cite{CG-Nash} on the above spending graph. The algorithm is given in \autoref{alg:marketRounding}, and the steps are explained below. 

By rearranging the spending of the agent, we may assume that $G$ is a forest of trees. Moreover, each tree in the forest must contain an agent-vertex since an item is assigned to at least one agent in $\xx$. 

Step 2 of the algorithm chooses an arbitrary agent-vertex from each tree to be the root of the tree. 
Steps 3 and 4 then assign any leaf-item and item with price less than 1/2 to its parent-agent. 

Step 5 of the algorithm computes the optimal matching of the remaining items to the agents, given the assignments that have been made in the previous steps. 
This can be done by computing a matching that maximizes sum of the logarithms of the valuations, which is equivalent to maximizing the product of the valuations.

\begin{figure}[ht]
	
	\begin{algbox}
		$\zz = \textsc{SpendingRestrictedRounding}\left( \pp, \xx \right)$
	            
		\textbf{Input:} Spending-restricted equilibrium price $\pp$, and the corresponding fractional allocation $\xx$. 
		
		\textbf{Output:} Integral allocation $\zz$.
		\begin{enumerate}
			\item Compute a spending graph $G$ from agents to items according to $\xx$.
			\item Choose a root-agent for each tree in the $G$. 
			\item Assign any leaf-item to the parent-agent.
			\item Assign any item $j$ of type $i$ with $p_i \leq 1/2$ to the parent-agent.
			\item Compute the optimal matching of the remaining items to the adjacent agents.
			\item Return the obtained integral allocation.
		\end{enumerate}
	\end{algbox}
	
	\caption{Algorithm for Rounding a Spending-Restricted Fractional Allocation.}
	
	\label{alg:marketRounding}
	
\end{figure}

\paragraph{Approximation Guarantee}
We show that the rounding algorithm gives a factor 2 approximation by constructing an instance of NSW problem under 
linear utilities that has the same upper bound as the one in \autoref{lem:upperbound}. 
A solution under linear utilities obtained by the rounding procedure in \autoref{alg:marketRounding} can then be viewed as a solution under SPLC utilities. 
Since the rounding procedure gives a solution under linear utilities that is at least a factor 2 of the upper bound in \autoref{lem:upperbound} (as shown in \cite{C+}),
it will follow that the same solution is also a factor 2 approximation under SPLC utilities.

\paragraph{An NSW Instance under Linear Utilities} Given a spending restricted allocation $\xx$ and corresponding prices $\pp$ under SPLC utilities, we build another instance 
$\mathcal{I}_{{linear}}$ of NSW under linear utilities. 
The instance has the same set of agents and items as the original instance $\mathcal{I}_{{splc}}$ under SPLC utilities. 
However, valuations have to be redefined as follows:

\begin{enumerate}
\item For each agent $a$, consider each item $(a,i,j)$ assigned (fractionally) to $a$.
\begin{enumerate}
	\item If it is a superior item, let it have a valuation of $p_i + q_{aij}$ for $a$ and 0 for all other agents. 
	\item If it is an active item, let it have a valuation of $p_i$ for $a$. 
\end{enumerate}
\item Set all remaining valuations to 0.
\end{enumerate} 

\paragraph{Spending-Restricted Model under Linear Utilities} Since a market under linear utilities is a special case of the one under SPLC utilities, our definition of spending-restricted allocation in Section~\ref{sec.model} also applies to the linear utilities case.  
To be precise, a spending-restricted equilibrium under linear utilities is defined as a fractional solution 
$\xx$ and a price vector $\pp$ such that every agent
spends all of his budget on his maximum bang-per-buck items at price $\pp$, and the total spending on each item $i$ is equal to $\min(1,p_i)$.
Also, \autoref{lem:upperbound} can be applied to the linear utilities case.

\begin{lemma}
\label{lem:reduction}
Let $\xx$ and $\pp$ be a spending-restricted solution for $\mathcal{I}_{splc}$. Then
$\xx$ is also a spending restricted solution for $\mathcal{I}_{linear}$. 
Moreover, \autoref{lem:upperbound} gives the same upper bound for both $\mathcal{I}_{{splc}}$ and $\mathcal{I}_{{linear}}$.
\end{lemma}
\begin{proof} First we give a spending-restricted equilibrium price $\bf{p'}$ for $\mathcal{I}_{linear}$ under $\xx$. 
Notice that by the way we construct  $\mathcal{I}_{linear}$, each item has the same valuation for all agents that it is assigned to in $\xx$.
We price the item at that valuation. 
It can easily be seen that under such $\bf{p'}$, each agent has bang-per-buck value of 1 and spends all his budget on maximum bang-per-buck items.
 
Since $\xx$ and $\pp$ is a spending-restricted solution for $\mathcal{I}_{splc}$, the base spending on each item of type $i$ is exactly $\min(1,p_i)$. 
Base spending only differs from total spending at items where the agents have to pay extra utility money. 
On each of those items, the total spending, and hence the price under $\bf{p'}$, must be less than 1 since the item is assigned completely to only one agent.  
It follows that total spending on each item $i$ is equal to $\min(1,p'_i)$ under $\xx$ and $\bf{p'}$. 

Finally, the sets of high-price items with respect to $\bf{p}$ and $\bf{p'}$ are identical. 
Therefore, the product in the RHS of the inequality in \autoref{lem:upperbound} has the same value in both cases. In other words, it can serve as 
an upper bound for both $\mathcal{I}_{{splc}}$ and $\mathcal{I}_{{linear}}$.
\end{proof}

%By carefully analyzing the rounding algorithm, \cite{C+} show that the approximation factor of the rounding algorithm in the case of linear utility NSW is 2. 
%Since Step 5 computes the optimal matching of the remaining items, it suffices to show that there is an assignment that obtains the desired approximation guarantee. 
%
%To show that, \cite{C+} first introduce the \emph{pruned} spending graph $P$ by removing some edges from $G$. 
%Specifically, for each item $j$ that has more than one child-agent in $G$, the edges connecting it to all child-agents are removed, except the one between $j$ and the child-agent that spends the most money on $j$. 
%
%Note that $P$ is also a forest of trees. For each tree $T$ in $P$, let $M_T$ be the union of items
%in $T$ with the items that were assigned to agents in $T$ in Steps 3 and 4. Also, let $H_T(\pp)$ be the set of items $j$ such that $t(j) \in H(\pp)$ where $t(j)$ is the type of $j$. \cite{C+} proved the following lemma:
%
%\begin{lemma} {\cite{C+}}
%\label{lem:trees}
%For any tree $T$ with $n_T$ agents, the allocation $\zz$ returned by the rounding algorithm satisfies
%\[
%	\prod_{a \in T} v_a(\zz) \geq \frac{1}{2^{n_T}} \prod_{j \in H_T(\pp)} p_{t(j)}.
%\]
%\end{lemma} 
	
%From lemma~\ref{lem:upperbound} and lemma~\ref{lem:trees}, 
%\[  \prod_{a} v_a(\zz)  = \prod_{T} \prod_{a \in T} v_a(\zz) \geq  \frac{1}{2^n} \prod_{i \in H(\pp)} p_i^{k_i} \geq \frac{1}{2^n} \prod_a u_a(\xx^*).  \]

By \cite{C+}, applying the rounding procedure in \autoref{alg:marketRounding} to $\xx$ gives an integral solution that is at least factor 2 of the upper bound of $\mathcal{I}_{{linear}}$ in \autoref{lem:upperbound}.
From \autoref{lem:reduction}, that upper bound is also an upper bound of $\mathcal{I}_{{splc}}$.
We can state the main theorem.

\begin{theorem}
$\textsc{SpendingRestrictedRounding}$ is a factor 2 approximation algorithm for NSW under SPLC utilities.
\end{theorem}

\def\cE{{\cal E}}
\section{Real Stable Polynomial Approach}
\label{sec.stability}

In this section we prove \autoref{thm:mainstable}.
First, in Section~\ref{sec:stableprelim} we give a short overview of stable polynomials and we discuss the main tool that we use in our proof. Then, in Section~\ref{sec:stableproof} we prove the theorem. 
 %algorithm for Nash welfare maximization when buyers have SPLC utilities using the machinery of stable polynomials. 

%Throughout this section we adopt the following notation. 
%Let $K:=\max_i k_i$. Without loss of generality we assume we have $K$ copies of each item. To justify this we assume that the utility of all agents from additional copies of each item is $0$.

For a vector $\by$, we write $\by > 0$ to denote that all coordinates of $\by$ are more than $0$.
For two vectors $\bx,\by\in \R^n$ we define $\bx\by=(x_1y_1,\dots,x_ny_n)$. Similarly, we define $\bx/\by=(x_1/y_1,\dots,x_n/y_n)$.
For a vector $\bx\in\R^n$,  we define $\exp(\bx):=(e^{x_1},\dots,e^{x_n})$. For two vectors $\bx, \by\in\R^n$ we define $\bx^\by$ as $\prod_{i=1}^n x_i^{y_i}$.
For a real number $c\in\R$ we write $c^\bx$ to denote $\prod_{i=1}^n c^{x_i}$.

%For vectors $\bx,\by\in \R^n$ we write 
%$$\bx^\by:=\prod_{i=1}^d x_i^{y_i}.$$

\subsection{Preliminaries}
\label{sec:stableprelim}
Stable polynomials are natural multivariate generalizations
of real-rooted univariate polynomials. For a complex number $z$, let
$\image(z)$ denote the imaginary part of $z$.
We say a polynomial $p(z_1,\dots,z_n)\in\C[z_1,\dots,z_n]$ is {\em stable}
if whenever $\image(z_i)>0$ for all $1\leq i\leq m$, $p(z_1,\dots,z_n)\neq 0$. As the only exception, we also call the zero polynomial stable. We say $p(.)$ is real stable, if it is stable and all of its coefficients are real. It is easy to see that any univariate polynomial is real stable  if and only if it is real rooted. 

For a polynomial $p$, let $\deg p$ be the maximum degree of all monomials of $p$.
We say a polynomial $p\in\R[z_1,\dots,z_n]$ is degree $k$-homogeneous, or $k$-homogeneous, if every monomial of $p$ has total degree exactly $k$. Equivalently, $p$ is $k$-homogeneous if for all $a\in\R$, we have
$$ p(a\cdot z_1,\dots,a\cdot z_n)=a^k p(z_1,\dots,z_n).$$
%We say a monomial $z_1^{\alpha_1}\dots z_n^{\alpha_n}$ is {\em square-free} if $\alpha_1,\dots,\alpha_n\in\{0,1\}$. 
%For a set $S\subset 2^{[n]}$ we write
%$$ \bz^S=\bz^{\bone_S}=\prod_{i\in S} z_i.$$
%Thus, we can represent a square-free monomial with the set of indices of variables of that monomial.
%For a multilinear polynomial $p\in \R[z_1,\dots,z_n]$ and a set $S\subseteq [n]$, we write $C_p(S)$ to denote the coefficient of the monomial $\bz^S$.

For a polynomial $p\in \R[z_1,\dots,z_n]$ and a vector $\bkappa\in \Z^n$, let $C_p(\bkappa)$ be the coefficient of the monomial $\prod_{i=1}^n z_i^{\kappa_i}$ in $p$. 

The following facts about real stable polynomials are well-known
\begin{fact}\label{fact:stablemult}
If $p,q\in\R[z_1,\dots,z_n]$ are real stable, then $p\cdot q$ is also real stable.
\end{fact}
\begin{fact}\label{fact:stablelinear}
For any nonnegative numbers $a_1,\dots,a_n$, the polynomial $a_1z_1+\dots+a_nz_n$ is real stable. 
\end{fact}

\noindent The following theorem is proved by Gurvits and was the key to the recent application of stable polynomials to the Nash welfare maximization problem \cite{AOSS17}.
\begin{theorem}[\cite{Gur06}] 
\label{thm:gurvits}
For any  $n$-homogeneous stable polynomial $p\in\R[z_1,\dots,z_n]$ with nonnegative coefficients, 
$$ C_p(1,1,\dots,1) \geq \frac{n!}{n^n} \inf_{\bz>0} \frac{p(z_1,\dots,z_n)}{z_1\dots z_n}.$$	
\end{theorem}
We use the following generalization of the above theorem which was recently proved in \cite{AO16}.

\begin{theorem}[\cite{AO16}]\label{thm:pqstability}
For any real stable  polynomials $p,q\in \R[z_1,\dots,z_n]$ with nonnegative coefficients, %if one of $p,q$ is multilinear, then 
\begin{eqnarray} e^{-\min\{\deg p,\deg q\}} \cdot \adjustlimits\sup_{\balpha} \inf_{\by,\bz>0} \frac{p(\by)q(\balpha\bz)}{(\by\bz)^\balpha}&\leq& \sum_{\bkappa\in\Z_+^n} \bkappa! C_p(\bkappa)C_q(\bkappa). \label{eq:pqlowerbound}\\
\adjustlimits\sup_{\balpha\geq 0}\inf_{\by,\bz>0} \frac{p(\by)q(\balpha\bz)}{(\by\bz)^\balpha} &\geq &
\sum_{\bkappa\in\Z_+^n} \bkappa^\bkappa C_p(\bkappa)C_q(\bkappa).\label{eq:pqupperbound}
\end{eqnarray}
where $\bkappa!:=\prod_{i=1}^n \kappa_i!$.
\end{theorem}
Furthermore, it was shown in \cite{AO16} that one can optimize $\sup_{\balpha\geq 0} \inf_{\by,\bz>0} \frac{p(\by)q(\balpha\bz)}{(\by\bz)^\balpha}$ using classical convex programming tools. Equivalently, it is enough to optimize the following convex function 
$$\adjustlimits\sup_{\balpha\geq 0} \inf_{\by,\bz>0} \log \frac{p(\exp(\by))q(\balpha\exp(\bz))}{e^{\langle \balpha,\by\rangle}e^{\langle \balpha,\bz\rangle}},$$
where $\langle \balpha,\by\rangle:=\sum_{i=1}^n \alpha_iy_i$.

\def\cS{{\cal S}}
\subsection{Main Proof}\label{sec:stableproof}
%Consider a vector $\by$ of $m\cdot n\cdot K$ variables, $y_{aij}$, corresponding to an agent $a$, and the $j$-th copy of item $i$.
In this part we prove \autoref{thm:mainstable}.
Our main tool is \autoref{thm:pqstability}. To use that, first we need to construct two real stable polynomials $p,q$. Then, we use \autoref{thm:pqstability} to write a convex relaxation for the Nash welfare objective with SPLC utilities. Finally, we will describe our rounding algorithm and prove the correctness.

%Let $K:=\max_i k_i$.
Let $T$ be the set of all triplets $(a,i,j)$ where $a\in [n]$, $i\in [m]$ and $j\in [k_i]$.
Let $\bx\in \R_+^{T}$ be a vector; ideally we would like $\bx$ to be an allocation vector.
For a vector $\bx$, let $p_\bx\in \R[y_1,\dots,y_m]$ be the following real stable polynomial:
$$ p_\bx(y_1,\dots,y_m)=\prod_{a=1}^n \left(\sum_{i=1}^m \sum_{j=1}^{k_i} x_{aij} u_{aij} y_i \right).$$
Observe that if $\bx$ is an integral allocation vector then $p_\bx(1,1,\dots,1)$ is the Nash welfare corresponding to $\bx$.
The polynomial $p_\bx$ is real stable since stable polynomials are closed under multiplication, 
 \autoref{fact:stablemult}, and  any linear function with nonnegative coefficients is real stable,  \autoref{fact:stablelinear}. 
It is not hard to see that $p_\bx$ is $n$-homogeneous and has nonnegative coefficients (since $u_{aij}\geq 0$). 

Let us characterize all possible monomials of $p_\bx$. 
For a set $S\subseteq T$ let $\be_S$ be the vector where for all $i\in [m]$, $\coord{e}{S}{i}$ denotes the number of triplets of the form $(.,i,.)$ in $S$, i.e.,
$$ \coord{e}{S}{i}:=|\{(a,i,j)\in S: a\in[n], j\in [k_i]\}|.$$
Let us describe the monomials of $p_\bx$. For every set $S\in \binom{T}{n}$ define $C(S)$ in the following way
\[ C(S):=\begin{cases}
	\prod_{(a,i,j)\in S} u_{aij} & \text{if $S$ has one element of the form $(a,.,.)$ for every $a\in[n]$},\\
	0 & \text{otherwise}.
\end{cases} \]
Abusing notation slightly, for every set $S\in \binom{T}{n}$, define $C_{p_\bx}(S)$ as follows:

 %For any $S\subseteq T$ of size $n$, $p_\bx$ has a monomial $\by^{\be_S}$ with coefficient
% $p_\bx$ has a monomial $\by^{\be_S}$ with coefficient
\[ C_{p_\bx}(S) := C(S)\prod_{(a,i,j)\in S}x_{aij}. \]
Then the following holds:
\[ p_\bx(\by)=\sum_{S\in\binom{T}{n}}C_{p_\bx}(S)\by^{\be_S}. \]
We remark that different sets $S$ can produce the same $\be_S$. So the above expression is not necessarily the standard way of writing a polynomial as a sum of monomials, i.e. the above monomials can be merged.

Note that if $\bx$ is a $\{0,1\}$ vector, then for any $S$ where $\coord{e}{S}{i}>k_i$ for some $i$, $C_{p_\bx}(S)=0$; In other words, for an integral $\bx$, the degree of $y_i$ in $p_\bx$ is at most $k_i$.
But if $\bx$ is not integral, this is not necessarily true. 
Ideally, we would like to avoid these bad sets because they may unboundedly increase the value $p_\bx$ for fractional allocation vectors.
We specifically choose a real stable polynomial $q$ such that the maximum degree of $y_i$ in $q$ is at most $k_i$.
%This way, $q(\be_S)p_\bx(\be_S)=0$ for any such unwanted sets.

Let $K:=\sum_{i=1}^m k_i$. Define the following real stable polynomial 
$$ q(\by)=\frac{1}{(K-n)!} \partial_t^{K-n}  \prod_{i=1}^m (t+y_i/k_i)^{k_i}\Big|_{t=0}%=\sum_{S\in {\cS}} \by^S \prod_{i=1}^m {K\choose S(i)} 
$$
In words, $q$ is equal to the coefficient of monomial $t^{K-n}$ in the polynomial $\prod_{i=1}^m (t+y_i/k_i)^{k_i}.$  Observe that by definition, the degree of $y_i$ in $q$ is  at most $k_i$. Furthermore, $q(\by)$ is $n$-homogeneous.

Let $\cS\subseteq 2^{T}$ be a family of subsets of $T$, consisting of all subsets $S$ where $|S|=n$ and $\coord{e}{S}{i}\leq k_i$ for all $i\in [m]$.
The following lemma is immediate from the above discussion.
\begin{lemma}\label{lem:pdotqcharacterization}
$$ \sum_{\bkappa\in \Z_+^n} \bkappa^\bkappa C_{p_\bx}(\bkappa) C_q(\bkappa)=\sum_{S\in \cS} {\be_S}^{\be_S} C_{p_\bx}(S)C_q(\be_S).$$
\end{lemma}

Next, we will use \autoref{thm:pqstability} on polynomials $p$ and $q$ to design our relaxation and approximation algorithms.
First, we show the following lemma.
We will then use it to write a convex relaxation of the optimum solution.
\begin{lemma} For any integral allocation vector $\bx$, 
\begin{equation}\label{eq:pqrelaxation}
 \adjustlimits\sup_{\balpha\geq 0} \inf_{\by,\bz>0} \frac{p_\bx(\by)q(\balpha\bz)}{(\by\bz)^\balpha} \geq \prod_a u_a(\bx).
 \end{equation}
 \end{lemma}
\begin{proof}
Since $p_\bx,q$ are real stable polynomials with nonnegative coefficients by \autoref{thm:pqstability}, \eqref{eq:pqupperbound}, we have
 $$\adjustlimits\sup_{\balpha\geq 0} \inf_{\by,\bz>0} \frac{p_\bx(\by)q(\bz)}{(\by\bz/\balpha)^\balpha} \geq \sum_{\bkappa\in \Z_+^n}\bkappa^\bkappa C_{p_\bx}(\bkappa)C_q(\bkappa)= \sum_{S\in\cS}  C_{p_{\bx}}(S)C_q(\be_S)  \prod_{i=1}^m {\coord{e}{S}{i}}^{\coord{e}{S}{i}}$$
 where we used \autoref{lem:pdotqcharacterization}.
% where we used that all monomials of $p_\bx$ are of the form $\by^{\be_S}$ for some $S\in \cS$.
 
 Now, let us calculate $C_q(\be_S)$. %Recall that $S$ has $|S_i|$ triplets of the form $(.,i,.)$. 
Observe that the coefficient of $y_i^{\coord{e}{S}{i}}t^{k_i-\coord{e}{S}{i}}$ in
$$ (t+\sum_{a,j} y_i/k_i)^{k_i}$$
is exactly equal to $k_i^{-\coord{e}{S}{i}}{k_i\choose \coord{e}{S}{i}}$.
Therefore, 
\begin{equation}\label{eq:Cq} 
C_q(\be_{S})=\prod_{i=1}^m k_i^{-\coord{e}{S}{i}}{k_i\choose \coord{e}{S}{i}}.
\end{equation}
%= \prod_{i=1}^m \frac{K!}{(K-|S_i|)!}.$$
Therefore
\begin{eqnarray*}
\sum_{S\in \cS} C_{p_\bx}(S) C_q(\be_S)\prod_{i=1}^m {\coord{e}{S}{i}}^{\coord{e}{S}{i}}&=& \sum_{S\in\cS} C_{p_\bx}(S) \prod_{i=1}^m k_i^{-\coord{e}{S}{i}} {\coord{e}{S}{i}}^{\coord{e}{S}{i}} {k_i\choose {\coord{e}{S}{i}}}\\
&=&\sum_{S\in \cS} C_{p_\bx}(S) \prod_{i=1}^m \prod_{j=0}^{{\coord{e}{S}{i}}-1} \frac{(k_i-j){\coord{e}{S}{i}}}{k_i (\coord{e}{S}{i}-j)}\\
&\geq & \sum_{S\in \cS} C_{p_\bx}(S) = p_\bx(1,1,\dots,1),
\end{eqnarray*}
where in the inequality we used that $\coord{e}{S}{i} \leq k_i$. 
 %= \sum_{\bkappa\in\cS}  C_{p_\bx}(\bkappa) \prod_{i=1}^m \frac{K!}{(K-\kappa_i)!}
Finally, to conclude the lemma, note that $p_\bx(\bone)=\prod_a u_a(\bx)$. 
\end{proof}
%In other words, for every monomial $\by^\bkappa$ of $p_\bx(\by)$, we have $\bkappa\in \cS$.
%This proves \eqref{eq:pqrelaxation}.

Next, we use \eqref{eq:pqrelaxation} to write a convex relaxation for the optimum solution.
	\begin{equation}
		\label{eq:convprog}
		\begin{aligned}
			\adjustlimits\sup_{\bx,\balpha\geq 0} \inf_{\by,\bz>0} ~~~ & \frac{p_\bx(\by)q(\balpha\bz)}{(\by\bz)^\balpha}, &\\
			\st ~~~ & \sum_{a=1}^{n}\sum_{j=1}^{k_i} x_{aij}\leq k_i & \forall i\\
			& x_{aij}\leq 1 & \forall a, i, j\\
		\end{aligned}
	\end{equation}
	It follows by \eqref{eq:pqrelaxation} that the above mathematical program is a relaxation of the optimum.
Furthermore, observe that we can turn the above program into an equivalent convex program by a change variable $\by\leftrightarrow \exp(\by)$ and $\bz\leftrightarrow \exp(\bz)$. This proves the first part of \autoref{thm:mainstable}.
%	And here is the definition of the set $C$:
%	\[ C := \{\alpha=(\alpha_1,\dots,\alpha_m)\in \MZ^m\mid \forall i, 0\leq \alpha_i\leq k_i~\wedge~ \sum_{i=1}^{m}\alpha_i=n \}. \]
%	This definition of $C$ makes sure that the convex program \eqref{eq:convprog} remains a relaxation. In other words if $x_{aij}$'s are integral (i.e. $0/1$-valued), then the definition of $C$ makes sure that the objective in \ref{eq:convprog} is at least as large as that of \ref{eq:defconvprog}.

Next, we describe our rounding algorithm. Let $\bx$ be an optimal solution of the convex program.	
W.l.o.g. we can assume that for all $i$, $\sum_{a=1}^{n}\sum_{j=1}^{k_i}x_{aij}=k_i$. 
For each $1\leq i\leq m$,  choose $k_i$ samples independently from all triplets of the form $(.,i,.)$, where $(a,i,j)$ is chosen with probability $x_{aij}/k_i$; if $(a,i,j)$ is sampled, assign one of the copies of item $i$ to agent $a$.
For each $(a,i,j)$, let  $X_{aij}$ be the random variable indicating that $(a,i,j)$ is sampled (at least once). 
%The details are described in \autoref{alg:rounding}.
\begin{algbox}
\begin{algorithmic}
%		\State Let $X_{aij}:=0$ for all $a,i,j$.
		\For{each item type $i$}
		\For{$t=1\to k_i$}
			\State Sample $a,j$ with probability $x_{aij}/k_i$. Assign one of the copies of item $i$ to agent $a$. %If $(a,i,j)$ is already sampled do nothing.
			%\If{$X_{a i j}=0$}
			%	\State $X_{a i j}\leftarrow 1$.
			%	\State Give a copy of item $i$ to agent $a$.
			%\EndIf
		\EndFor
		\EndFor
	\end{algorithmic}
	\label{alg:rounding}
	\end{algbox}
	
Observe that the utility of agent $a$ at the end of the rounding procedure is at least 
$$u_a(\bx)\geq \sum_{i=1}^{m}\sum_{j=1}^{k_i}X_{aij}u_{aij}.$$
Note that we have an inequality as opposed to equality because $a$ may only receive two copies of item $i$ because $(a,i,1)$ and $(a,i,3)$ are sampled; in this case we write $u_{ai1}+u_{ai3}$ in the above sum to denote the utility of $a$ from item $i$, whereas the true utility of $a$ from item $i$ is $u_{ai1}+u_{ai2} \geq u_{ai1}+u_{ai3}$.

Therefore the expected  Nash welfare of the rounding algorithm is at least
	\begin{equation}
		\label{eq:expectation}
		\Ex[\text{ALG}] \geq \Ex\left[\prod_{a=1}^{n}\left(\sum_{i=1}^{m}\sum_{j=1}^{k_i}X_{aij}u_{uij}\right)\right] = \sum_{S\in\cS} \Ex\Big[\prod_{(a,i,j)\in S} X_{aij}\Big] C(S)
	\end{equation}
%	where $\mF$ is a family of sets of triplets $(a, i, j)$. The family $\mF$ consists of subsets $S$ for which the following hold: $|S|=n$ and for each $i$, $S$ contains no more than $k_i$ triplets of the form $(\cdot, i, \cdot)$. The coefficient $c_S$ is defined as follows:
%	\[ c_S=\begin{cases}
%		0 & |\{a\mid \exists i, j: (a, i, j)\in S\}|<n,\\
%		\prod_{(a, i, j)\in S}u_{aij} & |\{a\mid \exists i, j: (a, i, j)\in S\}|=n.
%	\end{cases} \]
%	Note that the family $\mF$ is very closely related to the set of constraints put on $y_i$'s in the convex program \ref{eq:convprog}, and in particular the set $C$. One equivalent way of defining $\mF$ would be to say that it is the family of all sets $S$ of triplets such that $\sum_{(a,i,j)\in S} e_i\in C$ (where $e_i$ is the $i$-th element of the standard basis in $\MZ^m$). Let us name this vector $t_S$:
%	\[ t_S := \sum_{(a, i, j)\in S} e_i.\]
%	When $S$ is clear from the context we will drop the subscript $S$ and simply write $t$.	
The following key lemma lower bounds $\Ex\big[\prod_{(a,i,j)\in S} X_{aij}\big]$ for a given $S\in \cS$.
\begin{lemma}
		\label{lem:lb}
		For any set $S\in \cS$, 
		\[ \Ex\left[\prod_{(a, i, j)\in S} X_{aij}\right]\geq \prod_{(a,i,j)\in S} \frac{x_{aij}}{k_i} \left(\prod_{i=1}^m e^{-\coord{e}{S}{i}}\frac{k_i!}{(k_i-\coord{e}{S}{i})!}\right).\]
	\end{lemma}
	\begin{proof}
		Note that the rounding procedure is independent  for different item types. So it is enough to separate this inequality, and prove it for each item type. So let us fix an item type $i\in [m]$ and let $S_i=S\cap \{(., i, .)\}$. We will show that
		\begin{equation}\label{eq:algicaseprob} \Ex\left[\prod_{(a,i,j)\in S_i} X_{aij}\right]=\mathbb{P}[X_{aij}=1, \forall (a,i,j)\in S_i] \geq  e^{-\coord{e}{S}{i}}\frac{k_i!}{(k_i-\coord{e}{S}{i})!} \prod_{(a,i,j)\in S_i} \frac{x_{aij}}{k_i}.
		\end{equation}
If $X_{aij}=1$ for all $(a,i,j)\in S_i$, then we can define the function $t:S_i\to \{1,\dots, k_i\}$, 
where $t(a,i, j)$ represents the time that $(a, i, j)$ was first sampled. Note that $t$ is necessarily injective.

Now, for any injective function $t:S_i\to[k_i]$, consider the event $\cE_t$ defined in the following way: $(a,i,j)$ was sampled at time $t(a,i,j)$ for every $(a,i,j)\in S_i$ and at every other time $t'\notin t(S_i)$, the sampled triplet $(a',i,j')$ was not in $S$.

By definition, the events $\cE_t$ are disjoint for different functions $t$. Therefore
\[ \Ex\left[\prod_{(a,i,j)\in S_i} X_{aij}\right]\geq \sum_{t:S_i\to[k_i]\text{ injective}} \mathbb{P}[\cE_t]. \]

Let
$$ z=\sum_{(a,i,j)\in S_i} x_{aij}/k_i$$
be the probability that at any given time, a triplet $(a,i,j)\in S$ is sampled. Then, $\cE_t$ occurs with probability
$$ (1-z)^{k_i-\coord{e}{S}{i}} \prod_{(a,i,j)\in S_i} \frac{x_{aij}}{k_i} \geq (1-\coord{e}{S}{i}/k_i)^{k_i-\coord{e}{S}{i}} \prod_{(a,i,j)\in S} \frac{x_{aij}}{k_i}  \geq  e^{-\coord{e}{S}{i}}\prod_{(a,i,j)\in S_i} \frac{x_{aij}}{k_i}.$$
The first inequality uses that $z\leq \coord{e}{S}{i}/k_i$ and the last inequality uses that $(1-z/k)^{k-z} \geq e^{-z}$ for all $0\leq z\leq k$.
Equation \eqref{eq:algicaseprob} follows from the above and the fact that there are $k_i!/(k_i-\coord{e}{S}{i})!$ choices for $t:S_i\to [k_i]$. This completes the proof of the lemma.
	\end{proof}
Now, we are ready to finish the proof of \autoref{thm:mainstable}. We show that the expected Nash welfare of the rounded solution is at least $e^{-2n} \sup_{\balpha\geq 0} \inf_{\by,\bz>0} \frac{p_\bx(\by)q(\balpha\bz)}{(\by\bz)^\balpha}$.

It follows by the above lemma that the expected Nash welfare of the allocation of the rounded solution is at least
\begin{eqnarray} \Ex[\text{ALG}]&\geq& \sum_{S\in \cS} C(S)\prod_{(a,i,j)\in S} \frac{x_{aij}}{k_i}  \prod_{i=1}^m e^{-\coord{e}{S}{i}} \frac{k_i!}{ (k_i-\coord{e}{S}{i})!}\nonumber \\
&=& e^{-n}\sum_{S\in \cS} C_{p_\bx}(S)\prod_{i=1}^m \frac{k_i!}{k_i^{\coord{e}{S}{i}} (k_i-\coord{e}{S}{i})!}
\label{eq:exalglower}
\end{eqnarray}
	On the other hand, since $p_\bx,q$ are real stable with nonnegative coefficients and are $n$-homogeneous, by \autoref{thm:pqstability}, \eqref{eq:pqlowerbound}, we have
	$$ \adjustlimits\sup_{\balpha\geq 0} \inf_{\by,\bz>0} \frac{p_\bx(\by)q(\balpha\bz)}{(\by\bz)^{\balpha}}\leq e^n \sum_{S\in \cS} \be_S! \cdot C_{p_\bx}(S) C_q(\be_S).$$
	
	Therefore, by \eqref{eq:Cq} we can write,
	\begin{eqnarray*}
	\adjustlimits\sup_{\balpha\geq 0} \inf_{\by,\bz>0} \frac{p_\bx(\by)q(\balpha\bz)}{(\by\bz)^{\balpha}} &\leq& e^n \sum_{S\in \cS} \be_S! C_{p_\bx}(\be_S) C_q(\be_S)\\
	&= &
	e^n \sum_{S\in \cS} C_{p_\bx}(S) \prod_{i=1}^m \frac{\coord{e}{S}{i}!{k_i\choose \coord{e}{S}{i}}}{k_i^{\coord{e}{S}{i}}}\\
	&=& e^n\sum_{S\in \cS} C_{p_\bx}(S) \prod_{i=1}^m \frac{k_i!}{k_i^{\coord{e}{S}{i}} (k_i-\coord{e}{S}{i})!} \leq e^{2n} \Ex[\text{ALG}].
	\end{eqnarray*}
	The last inequality follows from \eqref{eq:exalglower}.

\bibliography{kelly} 
\bibliographystyle{alpha}

\end{document}